\documentclass[twocolumn,%,superscriptaddress,
showpacs,preprintnumbers,amsmath,amssymb]{revtex4}
\usepackage{hyperref}

\usepackage{amssymb,amsmath,amsthm}
\newtheorem{Thm}{Theorem}

\theoremstyle{definition}

\newcommand{\bra}[1]{{\left\langle #1 \right|}}
\newcommand{\ket}[1]{{\left| #1 \right\rangle}}

\newcommand{\C}{\mbox{$\mathbb C$}}

\newcommand{\T}{\mbox{$\mathrm{tr}$}}

\begin{document}
%%%%%%%%%%%%%%%%%%%%%%%%%%%%%%%%%%%%%%%%%%%%%%%%%%%%%%%%%%%%%%%%%%%%%%%%%%
%                                                                        %
%                                 Title                                  %
%                                                                        %
%%%%%%%%%%%%%%%%%%%%%%%%%%%%%%%%%%%%%%%%%%%%%%%%%%%%%%%%%%%%%%%%%%%%%%%%%%
\title{Polygamy of Entanglement in Multipartite Quantum Systems}

\author{Jeong San Kim}
\email{jkim@qis.ucalgary.ca}
\affiliation{
 Institute for Quantum Information Science,
 University of Calgary, Alberta T2N 1N4, Canada
}
\date{\today}

%%%%%%%%%%%%%%%%%%%%%%%%%%%%%%%%%%%%%%%%%%%%%%%%%%%%%%%%%%%%%%%%%%%%%%%%%%
%                                                                        %
%                              Abstract                                  %
%                                                                        %
%%%%%%%%%%%%%%%%%%%%%%%%%%%%%%%%%%%%%%%%%%%%%%%%%%%%%%%%%%%%%%%%%%%%%%%%%%
\begin{abstract}
We show that bipartite entanglement distribution (or entanglement of assistance) in multipartite
quantum systems is
by nature polygamous. We first provide an analytic upper bound for the concurrence of assistance
in bipartite quantum systems, and derive a polygamy inequality of multipartite
entanglement in arbitrary dimensional quantum systems.
\end{abstract}

%\keywords{}
\maketitle

Whereas quantum entanglement in bipartite quantum system has been intensively studied with
various applications, entanglement in multipartite quantum systems still seems far from
rich understanding.
One of the most distinct phenomena of quantum entanglement in multi-party systems is that
it cannot be freely shared among parties. For example, if two parties share a maximally entangled state,
they cannot have entanglement,
nor even classical correlations with any other party.
This is known as the {\em Monogamy of Entanglement} (MoE): The entanglement between one party and all others in
multipartite quantum systems bounds the sum of entanglement between one and each of the others.
MoE was shown to have a mathematical characterization in forms of inequalities in multi-qubit systems~\cite{ckw, ov}
using {\em concurrence}~\cite{ww} to quantify the shared entanglement among subsystems.

However, monogamy inequality using concurrence is know to fail in its generalization
for higher-dimensional quantum systems. In other words, the existence of quantum states violating
concurrence-based monogamy inequality was shown in higher-dimensional systems~\cite{ou, kds}.
Later, it was shown that those counterexamples of concurrence-based monogamy inequality still show
monogamous property of entanglement by using a different entanglement measure~\cite{kds}, and this
exposes the importance of having a proper way of quantifying entanglement.

Whereas, monogamy inequality is about the restricted sharability of multipartite entanglement,
entanglement distribution, which can be considered as a dual concept to the sharable entanglement,
is known to have a polygamous property (sometimes referred as {\em dual monogamy}) in multipartite quantum systems.
A mathematical characterization for the {\em Polygamy of Entanglement} (PoE) was first provided for multi-qubit
systems~\cite{gms, gbs} using {\em Concurrence of Assistance} (CoA)~\cite{lve}.
Recently, polygamy inequality was also shown in tripartite quantum system of arbitrary dimension using
{\em Entanglement of Assistance} (EoA) for the quantification of entanglement distribution~\cite{bgk}.
However, a general polygamy inequality of entanglement in multipartite higher-dimensional quantum system is still
an open question.

Here, we provide a strong clue for this question. We provide an analytical upper bound of CoA for arbitrary
bipartite mixed states, and derive a polygamy inequality of multipartite entanglement in arbitrary dimensional quantum systems.
The upper bound is saturated for any two-qubit states, and thus the derived polygamy inequality coincides with
the one proposed in~\cite{gbs} for multi-qubit systems.

For any bipartite pure state $\ket \phi_{AB}$,
%in a $d\otimes d'$ ($d\le d'$) quantum system,
its concurrence, $\mathcal{C}(\ket
\phi_{AB})$ is defined as~\cite{ww}
\begin{equation}
\mathcal{C}(\ket \phi_{AB})=\sqrt{2(1-\T\rho^2_A)},
\label{pure state concurrence}
\end{equation}
where $\rho_A=\T_B(\ket \phi_{AB}\bra \phi)$.
For any mixed state $\rho_{AB}$, its concurrence is defined as
\begin{equation}
\mathcal{C}(\rho_{AB})=\min \sum_k p_k \mathcal{C}({\ket {\phi_k}}_{AB}),
\label{mixed state concurrence}
\end{equation}
and its CoA is
\begin{equation}
\mathcal{C}^a(\rho_{AB})=\max \sum_k p_k \mathcal{C}({\ket {\phi_k}}_{AB}),
\label{CoA}
\end{equation}
where the minimum and maximum are taken over all possible pure state
decompositions, $\rho_{AB}=\sum_kp_k{\ket {\phi_k}}_{AB}\bra
\phi_k$.

For a three-qubit state $\ket{\psi}_{ABC}$, a polygamy inequality
of entanglement was first introduced as~\cite{gms}
\begin{equation}
\mathcal{C}_{A(BC)}^2\le(\mathcal{C}_{AB}^a)^2+(\mathcal{C}_{AC}^a)^2,
\label{3dual}
\end{equation}
where $\mathcal{C}_{A(BC)}=\mathcal{C}(\ket{\psi}_{A(BC)})$ is the concurrence of a 3-qubit state $\ket{\psi}_{A(BC)}$ for
a bipartite cut of subsystems between $A$ and $BC$ and $\mathcal{C}^a_{AB}=\mathcal{C}^a(\rho_{AB})$
with $\rho_{AB}=\T_{C}\left( \ket{\psi}_{ABC}\bra{\psi}\right)$.
Later, a generalization of Eq.~(\ref{3dual}) into $n$-qubit systems~\cite{gbs}
\begin{equation}
\mathcal{C}_{A_1 (A_2 \cdots A_n)}^2  \leq  (\mathcal{C}^a_{A_1 A_2})^2 +\cdots+(\mathcal{C}^a_{A_1 A_n})^2,
\label{ndual}
\end{equation}
was also introduced
for an arbitrary $n$-qubit pure state $\ket{\psi}_{A_1\cdots A_n} \in {\left(\C^2\right)}^{\otimes n}$.

Now, let us consider a bipartite pure state of arbitrary dimension
$\ket{\psi}_{AB}=\sum_{i=1}^{d_1}\sum_{k=1}^{d_2}a_{ik}\ket{ik}_{AB}$ in
%$\mathcal H^{A} \otimes \mathcal H^{B}$, where
$\mathcal H^{A} \otimes \mathcal H^{B} \simeq \C^{d_1}\otimes \C^{d_2}$.
In terms of the coefficients of $\ket
{\psi}_{AB}$,  its concurrence can also be expressed as~\cite{A}
\begin{align}
\mathcal{C}^2(\ket \psi_{AB})=&2(1-\T\rho^2_A)\nonumber\\
=&4\sum_{i<j}^{d_1}\sum_{k<l}^{d_2}|a_{ik}a_{jl}-a_{il}a_{jk}|^2,
\label{pure state concurrence2}
\end{align}
where $\rho_A=\T_B(\ket \psi_{AB}\bra \psi)$.

Let $m$ and $n$ be ordered pairs such that
\begin{equation}
m=(i,j),~ n=(k,l),~ i<j,~ k<l,
\label{pairs}
\end{equation}
with $i, j = 1, \cdots, d_1,~ k, l=1,\cdots, d_2$.
As $m$ has $D_1=d_1(d_1-1)/2$ choices of taking $i$ and $j$ from $d_1$ elements,
and similarly $D_2=d_2(d_2-1)/2$ choices for $n$, with some appropriate orderings of
$(i, j)$ and $(k, l)$, we can label $m$ and $n$ as
\begin{equation}
m=1,\cdots, D_1,~ n=1,\cdots, D_2.
\label{relabel}
\end{equation}
We also let
%$L_A^m$ and $L_B^n$ be
%generators of group $SO(d_1)$ and $SO(d_2)$ respectively, that is,
\begin{align}
L_A^m=&P_A^m\left(-\ket{i}_A\bra{j}+\ket{j}_A\bra{i}\right){P_A^m},\nonumber\\
L_B^n=&P_B^n\left(-\ket{k}_B\bra{l}+\ket{l}_B\bra{k}\right){P_B^n},
\label{sogroup}
\end{align}
where $P_A^m=\ket{i}_A\bra{i}+\ket{j}_A\bra{j}$ and $P_B^n=\ket{k}_B\bra{k}+\ket{l}_B\bra{l}$
are the projections onto the subspaces spanned by $\{ \ket{i}_A, \ket{j}_A \}$ and
$\{ \ket{k}_B, \ket{l}_B \}$ respectively.
By straightforward calculation, we have
\begin{equation}
|\langle\psi|(L_A^{m}\otimes L_B^{n})|\psi^*\rangle|^2=4|a_{ik}a_{jl}-a_{il}a_{jk}|^2,
\label{subconcurrence}
\end{equation}
%where $|\widetilde{\psi}_{mn}\rangle=(L_{m}\otimes L_{n})|\psi^*\rangle$,
and together with Eq.~(\ref{pure state concurrence2}), we have
\begin{equation}
\mathcal{C}^2(\ket \psi_{AB})=\sum_{m=1}^{D_1}\sum_{n=1}^{D_2}
|\langle\psi|(L_A^{m}\otimes L_B^{n})|\psi^*\rangle|^2.
\label{pureconcurrence3}
\end{equation}

Eq.~(\ref{subconcurrence}) can be considered as the squared concurrence of
the pure state (possibly unnormalized)
\begin{align}
\left(P_A^m \otimes P_B^n \right)\ket{\psi}_{AB}=&a_{ik}\ket{ik}_{AB}+a_{il}\ket{il}_{AB}\nonumber\\
&+a_{jk}\ket{jk}_{AB}+a_{jl}\ket{jl}_{AB}
\end{align}
in two-dimensional subspaces of $\mathcal H^{A}$ and $\mathcal H^{B}$ spanned by $\{ \ket{i}_A, \ket{j}_A \}$ and
$\{ \ket{k}_B, \ket{l}_B \}$ respectively.
%which is a two-qubit subspace.
Furthermore, Eq.~(\ref{pureconcurrence3}) implies that the concurrence of a bipartite pure state $\ket \psi_{AB}$
can be decomposed into the concurrences of two-qubit subspaces in Eq.~(\ref{subconcurrence}).

For a mixed state $\rho_{AB}=\sum_{i}p_i\ket{\psi_i}_{AB}\bra{\psi_i}=
\sum_{i}\ket{\xi_i}_{AB}\bra{\xi_i}$ with $\ket{\xi_i}_{AB}=\sqrt{p_i}\ket{\psi_i}_{AB}$,
its average concurrence is
\begin{align}
\sum_{i}p_i\mathcal C(\ket{\psi_i})=&\sum_{i}p_i \left(\sum_{m,n}
|\langle\psi_i|(L_A^{m}\otimes L_B^{n})|\psi_i^*\rangle|^2  \right)^{\frac{1}{2}}\nonumber\\
=&\sum_{i}\left(\sum_{m,n}
|\langle\xi_i|(L_A^{m}\otimes L_B^{n})|\xi_i^*\rangle|^2 \right)^{\frac{1}{2}}\nonumber\\
\leq& \sum_{i}\sum_{m,n}|\langle\xi_i|(L_A^{m}\otimes L_B^{n})|\xi_i^*\rangle|,
\label{ave}
\end{align}
and thus its CoA is
\begin{align}
\mathcal{C}^a\left(\rho_{AB}\right)=& \max\sum_{i}p_i\mathcal C(\ket{\psi_i})\nonumber\\
\leq& \max \sum_{i}\sum_{m,n}|\langle\xi_i|(L_A^{m}\otimes L_B^{n})|\xi_i^*\rangle|\nonumber\\
=&\sum_{m,n} \left( \max \sum_{i}|\langle\xi_i|(L_A^{m}\otimes L_B^{n})|\xi_i^*\rangle|\right),
\label{uppercoa1}
\end{align}
where the maxima are taken over all possible pure state decompositions of $\rho_{AB}$.
Here, we note that the term after the maximum in the last line of Eq.~(\ref{uppercoa1})
is the average concurrence of the state (possibly unnormalized)
\begin{equation}
\left(\rho_{AB}\right)_{mn}=\left(P_A^m \otimes P_B^n \right)\rho_{AB}\left(P_A^m \otimes P_B^n \right)
\end{equation}
in the two-qubit subspace spanned by  $\{ \ket{i}_A, \ket{j}_A \}$ and
$\{ \ket{k}_B, \ket{l}_B \}$.
Thus, the maximum value
\begin{equation}
\max \sum_{i}|\langle\xi_i|(L_A^{m}\otimes L_B^{n})|\xi_i^*\rangle|
\label{subCoA}
\end{equation}
can be considered as the CoA of the two-qubit state $\left(\rho_{AB}\right)_{mn}$;
therefore, by the optimization methods for CoA in two-qubit systems~\cite{lve}, we have
\begin{align}
\mathcal{C}^a \left(\left(\rho_{AB}\right)_{mn} \right)=&\max \sum_{i}|\langle\xi_i|(L_A^{m}\otimes L_B^{n})|\xi_i^*\rangle|\nonumber\\
%=&\T\sqrt{\sqrt{\rho_{AB}} \left(\widetilde{\rho}_{AB}\right)_{mn}  \sqrt{\rho_{AB}}}\nonumber\\
=&\mathcal F \left[\rho_{AB}, \left(\widetilde{\rho}_{AB}\right)_{mn} \right],
\label{2CoA}
\end{align}
where $\left(\widetilde{\rho}_{AB}\right)_{mn}=\left(L_A^m \otimes L_B^n\right)
\rho_{AB}^*\left(L_A^m \otimes L_B^n\right)$, and
$\mathcal F \left[\rho_{AB}, \left(\widetilde{\rho}_{AB}\right)_{mn} \right]$ is the
fidelity of $\rho_{AB}$ and $\left(\widetilde{\rho}_{AB}\right)_{mn}$ defined as
\begin{equation}
\mathcal F \left[\rho_{AB}, \left(\widetilde{\rho}_{AB}\right)_{mn} \right]
=\T\sqrt{\sqrt{\rho_{AB}} \left(\widetilde{\rho}_{AB}\right)_{mn}  \sqrt{\rho_{AB}}}.
\label{fidelity}
\end{equation}

Now, we are ready to have the following theorem.
\begin{Thm}
For any bipartite state $\rho_{AB} \in \mathcal B \left(\C^{d_1}\otimes \C^{d_2}\right)$,
\begin{align}
\mathcal C^a(\rho_{AB}) \leq&\sum_{m=1}^{D_1}\sum_{n=1}^{D_2}
\mathcal F \left[\rho_{AB}, \left(\widetilde{\rho}_{AB}\right)_{mn} \right]\nonumber\\
:=& \tau^a(\rho_{AB}),
\label{eq: CoAupper}
\end{align}
\label{Thm: CoAupper}
where $D_1=d_1(d_1-1)/2$, $D_2=d_2(d_2-1)/2$, and
$\left(\widetilde{\rho}_{AB}\right)_{mn}=\left(L_A^m \otimes L_B^n\right)
\rho_{AB}^*\left(L_A^m \otimes L_B^n\right)$.
\end{Thm}

For any bipartite mixed state $\rho_{AB}$ of arbitrary dimension, the sum of concurrences of all
possible two-qubit subspaces is known to provide a lower bound of the concurrence
$\mathcal C \left(\rho_{AB} \right)$~\cite{off}. By using the lower bound of the concurrence,
it was also shown that there is a proper monogamy inequality of entanglement in multipartite
quantum systems~\cite{off}.

Theorem~\ref{Thm: CoAupper} implies that the sum of CoA of $\left(\rho_{AB}\right)_{mn}$ from all possible
two-qubit subspaces of $\rho_{AB}$ forms an upper bound of $\mathcal C^a(\rho_{AB})$.
It can be also directly checked that this bound is saturated for any pure state and
two-qubit mixed state.

Now, we provide a polygamy inequality of multipartite entanglement
of arbitrary dimension in terms of the upper bound proposed in Theorem~\ref{Thm: CoAupper}.

\begin{Thm}
For any multipartite pure state $\ket{\psi}_{A_1\cdots A_n}$ in $\C^{d_1}\otimes \cdots\otimes \C^{d_n}$,
\begin{equation}
\left(\tau^a_{A_1(A_2\cdots A_n)}\right)^2 \leq \left(\tau^a_{A_1 A_2}\right)^2+\cdots +\left(\tau^a_{A_1 A_n}\right)^2.
\label{eq: npoly}
\end{equation}
where $\tau^a_{A_1(A_2\cdots A_n)}=\tau^a\left(\ket{\psi}_{A_1(A_2\cdots A_n)}\right)$
with respect to the bipartite cut $A_1-A_2\cdots A_n$, $\tau^a_{A_1A_k}=\tau^a(\rho_{A_1A_k})$
and $\rho_{A_1A_k}$ is the reduced density matrix of $\ket{\psi}_{A_1\cdots A_n}$ onto subsystem
$A_1A_k$ for $k=2, \ldots n$.
\label{Thm: npoly}
\end{Thm}

\begin{proof}
For $k=1,\ldots, n$, let $m_k=(i_k, j_k)$ be an ordered pair of $i_k, j_k \in \{ 1, \ldots, d_k\}$
such that $i_k < j_k$, and  let $M=\left(m_2, \ldots, m_n \right)$ be an $(n-1)$-tuple of the ordered pairs.
By letting $D_k=d_k(d_k-1)/2$ for $k=1,\ldots, n$, we have
\begin{align}
\left(\tau^a_{A_1(A_2\cdots A_n)}\right)^2=&\mathcal {C}^2_{A_1(A_2\cdots A_n)}\nonumber\\
=&\sum_{m_1}^{D_1}\sum_{M=1}^{D_2\cdots D_n}
\left[\left(\mathcal{C}_{A_1(A_2\cdots A_n)}\right)_{m_1M}\right]^2
\label{nconcurrence}
\end{align}
where
$\left(\mathcal{C}_{A_1(A_2\cdots A_n)}\right)_{m_1M}$ is the concurrence of the (unnormalized) state
$\left(\ket{\psi}_{A_1(A_2\cdots A_n)}\right)_{m_1M}$ in the subspace spanned by
$\{\ket{i_1}_{A_1}, \ket{j_1}_{A_1}\},  \ldots , \{\ket{i_n}_{A_n}, \ket{j_n}_{A_n}\}$.

As $\left(\ket{\psi}_{A_1(A_2\cdots A_n)}\right)_{m_1M}$ is an $n$-qubit (unnormalized) state for each $m_1$ and $M$,
it satisfies the multi-qubit polygamy inequality in Eq.~(\ref{ndual}), that is,
\begin{widetext}
\begin{align}
\left[\left(\mathcal{C}_{A_1(A_2\cdots A_n)}\right)_{m_1M}\right]^2\leq
\left[\left(\mathcal{C}^a_{A_1A_2}\right)_{m_1m_2}\right]^2 %\nonumber\\
+\cdots + \left[\left(\mathcal{C}^a_{A_1A_n}\right)_{m_1m_n}\right]^2,
\label{m1Mpoly}
\end{align}
\end{widetext}
where $\left(\mathcal{C}^a_{A_1A_k}\right)_{m_1m_k}$ is CoA of $\left(\rho_{A_1A_k}\right)_{m_1m_k}$,
the reduced operator of $\left(\ket{\psi}_{A_1(A_2\cdots A_n)}\right)_{m_1M}$ onto
subsystems $A_1A_k$ for $k=2,\ldots, n$.
From Eqs.~(\ref{nconcurrence}) and~(\ref{m1Mpoly}), we have
\begin{widetext}
\begin{align}
\left(\tau^a_{A_1(A_2\cdots A_n)}\right)^2\leq& \sum_{m_1=1}^{D_1}\sum_{m_2=1}^{D_2}\left[\left(\mathcal{C}^a_{A_1A_2}\right)_{m_1m_2}\right]^2
+\cdots+\sum_{m_1=1}^{D_1}\sum_{m_n=1}^{D_n}\left[\left(\mathcal{C}^a_{A_1A_n}\right)_{m_1m_n}\right]^2\nonumber\\
\leq& \left[\sum_{m_1=1}^{D_1}\sum_{m_2=1}^{D_2}\left(\mathcal{C}^a_{A_1A_2}\right)_{m_1m_2}\right]^2
+\cdots+\left[\sum_{m_1=1}^{D_1}\sum_{m_n=1}^{D_n}\left(\mathcal{C}^a_{A_1A_n}\right)_{m_1m_n}\right]^2\nonumber\\
=& \left(\tau^a_{A_1 A_2}\right)^2+\cdots +\left(\tau^a_{A_1 A_n}\right)^2,
\end{align}
\end{widetext}
where the last equation is due to Eq.~(\ref{2CoA}) and the definition of $\tau^a$.
\end{proof}

As the upper bound of CoA in Theorem~\ref{Thm: CoAupper} is saturated for any two-qubit mixed state,
Eq.~(\ref{eq: npoly}) in Theorem~\ref{Thm: npoly} is reduced to Eq.~(\ref{ndual}) for the case of multi-qubit systems.
Moreover, it can be easily seen that Eq.~(\ref{eq: npoly}) is saturated by $n$-qubit generalized
W-class states~\cite{ks}, that is,
\begin{equation}
\ket{W}_{A_1\cdots A_n} =a_1 \ket{1\cdots0}_{A_1\cdots A_n}+\cdots+a_{n} \ket{0\cdots1}_{A_1\cdots A_n},
\label{nWclass}
\end{equation}
with $\sum_{i=1}^{n}|a_i|^2 =1$.

To summarize, we have shown the polygamous nature of
distributed entanglement in multipartite quantum systems of arbitrary
dimension. By providing an analytic upper bound of CoA for arbitrary bipartite quantum states,
we have derived a polygamy inequality of entanglement in terms of the upper bound.
This upper bound is saturated for any two-qubit state, and thus the polygamy inequality proposed here
can be considered as a generalization of the result in~\cite{gbs} into higher-dimensional quantum systems.

One of the main difficulties in the study of multipartite entanglement
is that there can be several inequivalent classes that are not convertible to each other
under {\em Stochastic Local operations and classical communications} (SLOCC)~\cite{DVC}.
These inequivalent classes makes us hardly have an universal way of quantifying multipartite entanglement,
even in an abstract sense.

However, the existence of inequivalent classes of multipartite entanglement also reveals the
different characters among different classes. For example, three-qubit
systems are known to have two inequivalent classes of genuine tripartite entanglement:
one is the Greenberger-Horne-Zeilinger (GHZ) class~\cite{GHZ},
and the other one is the W-class~\cite{DVC}.
Although the inequivalentness of the classes is due to SLOCC convertibility~\cite{DVC},
these two classes also show extreme differences in terms of monogamy or polygamy inequalities of entanglement.
In other words, monogamy and polygamy inequalities are saturated by W-class states, whereas the differences between
terms in the inequalities can assume their largest values for the GHZ-class state.
Thus, MoE and PoE are not only just distinct phenomena in multipartite quantum systems, they also
provide us an efficient way of qualifying multipartite entanglement.

Our result is, we believe, the first case where the polygamy nature of multipartite entanglement
in arbitrary-dimensional quantum systems is discussed with mathematical characterizations.
Noting the importance of the study on high-dimensional multipartite entanglement,
our result can provide a rich reference for future work on the study of multipartite entanglement.

\section*{Acknowledgments}
This work is supported by iCORE, MITACS and USARO.
%%%%%%%%%%%%%%%%%%%%%%%%%%%%%%%%%%%%%%%%%%%%%%%%%%%%%%%%%%%%%%%%%%%%%%%%


\begin{thebibliography}{1}

\bibitem{ckw}
V. Coffman, J. Kundu and W. K. Wootters,
Phys. Rev. A {\bf 61}, 052306 (2000).

\bibitem{ov}
T. Osborne and F. Verstraete,
Phys. Rev. Lett. {\bf 96}, 220503 (2006).

\bibitem{ww}
W. K. Wootters,
Phys. Rev. Lett. {\bf 80}, 2245 (1998).

\bibitem{ou}
Y. C. Ou,
Phys. Rev. A {\bf 75}, 034305 (2007).

\bibitem{kds}
J. S. Kim, A. Das and B. C. Sanders,
Phys. Rev. A {\bf 79}, 012329 (2009).

\bibitem{gms}
G. Gour, D. Meyer and B. C. Sanders, Phys. Rev. A  {\bf 72},
042329 (2005).

\bibitem{gbs}
G. Gour, S. Bandyopadhay and B. C. Sanders, J. Math. Phys. {\bf
48}, 012108 (2007).

\bibitem{lve}
T. Laustsen, F. Verstraete and S. J. van Enk, Quantum Inf. Comput.
{\bf 3}, 64 (2003).

\bibitem{bgk}
F. Buscemi, G. Gour and J. S. Kim,
arXiv:0903.4413 (2009).

\bibitem{A}
S. J. Akhtarshenas,
J. Phys. A {\bf 38}, 6777 (2005).

\bibitem{off}
Y. C. Ou, H. Fan and S. M. Fei,
Phys, Rev. A {\bf 78}, 012311 (2008).

\bibitem{ks}
J. S. Kim and B. C. Sanders,
J. Phys. A {\bf 41}, 495301 (2008).

\bibitem{DVC}
W.~D\"{u}r, G. Vidal and J. I. Cirac,
%{\em   Three qubits can be entangled in two inequivalent ways},
Phys. Rev. A {\bf 62}, 062314 (2000).

\bibitem{GHZ}
D. M. Greenberger, M. A. Horne and A. Zeilinger,
{\em Bell's Theorem, Quantum Theory, and Conceptions of the Universe},
edited by M. Kafatos (Kluwer, Dordrecht, 1989), p. 69.

\end{thebibliography}
\end{document}